\documentclass[11pt,reqno]{amsart}
\usepackage[a4paper,margin=27mm,headheight=15pt]{geometry}
\usepackage[T1]{fontenc}
\usepackage[utf8]{inputenc}
\usepackage{libertinus}
\usepackage{amssymb,mathtools}
\usepackage{thmtools}
\usepackage{thm-restate}
\usepackage{bm}
\usepackage{booktabs,longtable,tabularx,array}
\usepackage{enumitem}
\usepackage[dvipsnames]{xcolor}
\usepackage[most]{tcolorbox}
\usepackage{tikz}
\usetikzlibrary{arrows.meta,positioning,calc}
\usepackage{microtype}
\usepackage{hyperref}
\usepackage{bookmark}

\hypersetup{
  colorlinks=true,
  linkcolor=MidnightBlue,
  citecolor=ForestGreen,
  urlcolor=BrickRed,
}

\allowdisplaybreaks[3]
\setlist[itemize]{leftmargin=1.6em,itemsep=0.25em,topsep=0.35em}
\setlist[enumerate]{leftmargin=1.8em,itemsep=0.25em,topsep=0.35em}
\renewcommand{\arraystretch}{1.18}
\newcommand{\C}{\mathbb C}
\newcommand{\tr}{\operatorname{tr}}
\newcommand{\rank}{\operatorname{rank}}

\newcommand{\ket}[1]{\lvert #1\rangle}
\newcommand{\bra}[1]{\langle #1\rvert}

\declaretheorem[name=Theorem,numberwithin=section]{theorem}
\declaretheorem[name=Corollary,sibling=theorem]{corollary}
\declaretheorem[name=Lemma,sibling=theorem]{lemma}

\declaretheorem[style=remark,name=Remark,sibling=theorem]{remark}

\newcommand{\isnew}{\texorpdfstring{{\normalfont\bfseries[here]}}{[New]}}

\title[On the two-copy distillability of Werner states and a new partial trace inequality]
      {On the two-copy distillability of Werner states \\ and a new partial trace inequality}

\date{\today}

\author{Thomas C. Fraser$^1$}

\address{$^1$
Department of Mathematical Sciences,University of Copenhagen, Copenhagen, 2100, Denmark}

\email{tcf@math.ku.dk}

\author{Felix Huber$^2$}
\address{$^2$
Division of Quantum Information,
Institute of Informatics,
Faculty of Mathematics, Physics and Informatics,
University of Gdańsk,
Wita Stwosza 57, 80-308 Gdańsk, Poland
}
\email{felix.huber@ug.edu.pl}

\author{Bal\'azs Pozsgay$^3$}
\address{
$^3$MTA-ELTE "Momentum" Integrable Quantum Dynamics Research Group,
ELTE Eötvös Loránd University,
Pázmány Péter sétány 1/A, 1117 Budapest, Hungary}

\author{Istv\'an Vona$^{3,4}$}
\address{$^4$Holographic Quantum Field Theory Research Group,
HUN-REN Wigner Research Centre for Physics,
Konkoly-Thege Miklós u. 29-33, 1121 Budapest, Hungary
}
\email{vona.istvan@wigner.hun-ren.hu}


\begin{document}

\begin{abstract}
Problem 5 in {\it Five Open Problems in Quantum Information Theory} [PRX Quantum 3, 010101 (2022)], asks whether the two-ququart Werner state $\varrho(4,-\tfrac12)$ is two-copy distillable, where $\varrho(d,\alpha)=(I+\alpha F)/(d^2+\alpha d)$. We answer it in the negative. To this end, we show the following stronger statement: for all $C\in M_{d_1d_2}(\mathbb{C})$ of rank at most $r \le d_1 d_2$,
$\mathrm{tr}_1(C)\|_F^2+\|\mathrm{tr}_2(C)\|_F^2 \le r\|C\|_F^2+\frac{1}{r}|\mathrm{tr}(C)|^2$. A result by Costa Rico on the equivalence of this inequality with two-copy undistillability at $r = 2$ then settles Problem 5: $\varrho(4,-\tfrac{1}{2})$ is {\em not} two-copy distillable. Furthermore, we show that $\varrho(d,\alpha)$ is two-copy undistillable for every $d\ge2$, if and only if $\alpha\ge-\tfrac{1}{2}$. Thus, the one and two-copy distillability regions of $\varrho(d,\alpha)$ coincide.
These results have been found and written up with AI tools, pointing towards a structural change affecting the field of quantum information and computation.
\end{abstract}

\thanks{
We thank Karol \.{Z}yczkowski for fruitful discussions.
TCF acknowledges financial support from VILLUM FONDEN via the QMATH Centre of Excellence (Grant No. 10059) and the Novo Nordisk Foundation (grant NNF20OC0059939 “Quantum for Life”).
FH was funded in whole or in part by the National Science Centre, Poland 2024/54/E/ST2/00451
and by the Polish National Agency for Academic Exchange under the Strategic Partnership Programme grant BNI/PST/2023/1/00013/U/00001.
For the purpose of Open Access,
the author has applied a CC-BY public copyright licence to any Author Accepted Manuscript (AAM) version arising from this submission.
BP and IV were supported by the Hungarian National Research, Development
and Innovation Office, NKFIH Grant No. K-145904.
}

\maketitle

\section{Introduction}\label{sec:starting-point}

\subsection{Werner states}
Let \((e_i)_{i\in[d]}\) be the standard basis of \(\C^d\), and let
\(E_{ij}:=e_ie_j^*\in M_d(\C)\) be the associated matrix units, so that
\((E_{ij})_{i,j\in[d]}\) is an orthonormal basis of \(M_d(\C)\) for the
inner product~\eqref{eq:ip-mat}: \(\langle E_{ij},E_{kl}\rangle_F
=\delta_{ik}\delta_{jl}\).  Let \(F\) denote the flip on \(\C^d\otimes\C^d\),
defined by \(F(x\otimes y)=y\otimes x\) and given by
\begin{equation}\label{eq:flip-def}
F=\sum_{i,j=1}^d E_{ij}\otimes E_{ji}\,.
\end{equation}
The operator $F$ is Hermitian and satisfies \(F^2=I\).  For \(\alpha\in[-1,1]\) the
\emph{Werner state} with parameter \(\alpha\) is
\begin{equation}\label{eq:werner-def}
\varrho(d,\alpha):=\frac{I_{d^2}+\alpha F}{d^2+\alpha d}\,,
\end{equation}
following the convention of~\cite[Eq.~(4)]{Pankowski2010}; \(\varrho_\alpha\) is
written for \(\varrho(d,\alpha)\) when \(d\) is clear.

It is known that the partial transpose \(\varrho(d,\alpha)^{T_B}\) is not positive semidefinite, i.e. has negative partial transpose (NPT),
if and only if
\(\alpha<-\tfrac{1}{d}\)~\cite{Werner1989}.
Thus, \(\varrho(4,-\tfrac12)\) is NPT.
For \(\alpha\ge-\tfrac{1}{d}\) the state is separable~\cite{Werner1989}.

\subsection{Entanglement distillation}

A bipartite state
\(\varrho\) is \emph{two-copy distillable} if some vector \(\psi\) of Schmidt rank at
most two across \(A_1A_2:B_1B_2\) satisfies
\(\bra{\psi}(\varrho^{T_B})^{\otimes2}\ket{\psi}<0\), where \(T_B\) is partial
transposition on the second subsystem of each copy~\cite{Horodecki1998, Pankowski2010}, and
\emph{two-copy undistillable} otherwise.

A state with positive partial transpose cannot be distilled: if
\(\varrho^{T_B}\succeq0\), then \((\varrho^{T_B})^{\otimes n}\succeq0\) for every
\(n\), so \(\bra{\psi}(\varrho^{T_B})^{\otimes n}\ket{\psi}\ge0\) for every \(\psi\)
and \(\varrho\) is undistillable at every number of copies.  Thus non-positivity of
the partial transpose (NPT)---the Peres criterion~\cite{Peres1996}---is {\it necessary}
for distillability~\cite{Horodecki1998}.
A {\it sufficient} condition is the reduction criterion:
any state violating the inequality
\(\varrho_A\otimes I-\varrho\succeq0\) or
\(I\otimes\varrho_B-\varrho\succeq0\), where \(\varrho_A=\tr_2(\varrho)\) and
\(\varrho_B=\tr_1(\varrho)\), is distillable~\cite{HorodeckiReduction1999}.

Between
the two lies a partial unresolved regime: whether \emph{every} NPT state is distillable is a longstanding open question.
Recall that \(\varrho(d,\alpha)\) is NPT if and only if
\(\alpha<-\tfrac{1}{d}\)~\cite{Werner1989}.
Furthermore, Werner states are one-copy undistillable if and only if
\(\alpha\ge-\tfrac12\)~\cite{Horodecki1998,DurCiracLewensteinBruss2000,CostaRicoWolf2025}.
Since one-copy
distillability trivially implies $k$-copy for every $k \geq 1$,
\(\varrho(d,\alpha)\) is in particular
two-copy distillable for every
\(\alpha<-\tfrac12\).

Regarding undistillability, Vianna and Doherty showed by numerical
computation that qutrit Werner states are two-copy undistillable on the whole
range \(-\tfrac12\le\alpha<-\tfrac13\)~\cite{Vianna2006}.
For general \(d\), Costa Rico and Wolf showed that Werner states are two-copy
undistillable for all \(\alpha\ge-\tfrac13\)~\cite{CostaRicoWolf2025}.

\bigskip
\begin{center}
\footnotesize
\renewcommand*{\arraystretch}{1.3}
\begin{tabular}{@{}l c c c c@{}}
\toprule
 & \(\alpha<-\tfrac12\) & \(-\tfrac12\le\alpha<-\tfrac13\) & \(-\tfrac13\le\alpha<-\tfrac{1}{d}\) & \(\alpha\ge-\tfrac{1}{d}\)\\
 & & (\(d\ge3\) only) & (\(d\ge4\) only) & \\
\midrule
PPT~\cite{Peres1996} & N & N & N & Y\\
NPT~\cite{Peres1996} & Y & Y & Y & N\\
1-copy distillable~\cite{Horodecki1998,DurCiracLewensteinBruss2000,CostaRicoWolf2025} & Y & N & N & N\\
2-copy distillable & Y & N {\bf [here]}\(^{\dagger}\) & N~\cite{CostaRicoWolf2025}\(^{\ddagger}\)
& N~\cite{Werner1989}\\
\bottomrule
\end{tabular}

\smallskip
{\footnotesize
\(^{\dagger}\)New for \(d\ge4\); the qutrit case \(d=3\) is due to Vianna and
Doherty~\cite{Vianna2006}.
\\\(^{\ddagger}\) Vacuous for \(d=3\), where this range contains no NPT states.}
\end{center}

\subsection{The open problem}
The concrete question addressed here is the following, listed as Problem~5 in
Horodecki, Rudnicki and \.Zyczkowski's \emph{Five Open Problems in Quantum
Information Theory}~\cite{FiveOpen2022}.

\begin{restatable}[Problem~5 of~\cite{FiveOpen2022}]{problem}{probFive}\label{prob:five}
Decide whether the Werner state \(\varrho(4,-\tfrac12)\) of two ququarts, \(d=4\),
defined in Eq.~\eqref{eq:werner-def}, is two-copy distillable.
\end{restatable}

The answer given in this work is negative [shown in Cor.~B].
More generally, we ask:

\begin{restatable}{problem}{probGeneral}\label{prob:five-point-one}
For which \(d\) and \(\alpha\) is \(\varrho(d,\alpha)\) two-copy distillable?
\end{restatable}

We show that for every \(d \ge 2\),
the Werner state \(\varrho(d, \alpha)\) is
two-copy undistillable if and only if \(\alpha \ge  - \tfrac{1}{2}\) [shown in Cor.~C].
Thus every Werner state is two-copy undistillable if and only if it is one-copy undistillable.

\bigskip
{\it The proof presented in this work was found with ChatGPT Sol-5.6
and refined with Claude Opus and Fable.
All results have been carefully checked by the authors.}
\bigskip

{\it
During preparation of this manuscript
we became aware of a similar work by Fu et al,
\href{https://arxiv.org/abs/2607.21367}{arXiv:2607.21367}
and a manuscript in preparation by
Kishor Bharti, Rishikesh Gajjala, and Tobias Haug.
They seem to use a different proof strategy but also use a variant of Lemma~\ref{lem:esidentity}.
}

\section{Undistillability from a partial trace inequality}
To answer Problems~\ref{prob:five} and \ref{prob:five-point-one}, we will make use of the following partial trace inequality.

\begin{theorem}[Partial-trace criterion; Costa Rico~{\cite[Theorem~1]{CostaRico2025}}, in the norm form of Costa Rico and Wolf~{\cite[Eq.~(31)]{CostaRicoWolf2025}}]\label{thm:criterion-alt}
Let \(d\ge2\) and \(\alpha\in[-1,0)\).  The Werner state \(\varrho(d,\alpha)\) is
two-copy undistillable if and only if every \(C\in M_{d^2}(\C)\) with
\(\rank(C)\le2\) satisfies
\begin{equation}\label{eq:crw-criterion}
\|\tr_1(C)\|_F^2+\|\tr_2(C)\|_F^2
\le\frac1{|\alpha|}\|C\|_F^2+|\alpha|\,|\tr(C)|^2\,.
\end{equation}
\end{theorem}

Costa Rico~\cite{CostaRico2025} states the criterion for \(n\) copies and in the
equivalent quadratic form \(q_\alpha(C)\ge0\), where
\(q_\alpha(C)=\|C\|_F^2+\alpha(\|\tr_1(C)\|_F^2+\|\tr_2(C)\|_F^2)+\alpha^2|\tr(C)|^2\);
only \(n=2\) is used here.  Writing \(\alpha=-t\) with \(t>0\) and dividing
\(q_{-t}(C)\ge0\) by \(t\) gives the norm form~\eqref{eq:crw-criterion} stated
above, which is the form recorded by Costa Rico and
Wolf~\cite[Eq.~(31)]{CostaRicoWolf2025} and the one applied here.

Costa Rico and Wolf~\cite{CostaRicoWolf2025} prove that~\eqref{eq:crw-criterion} holds for every rank-two \(C\) when
\(\alpha\ge -\tfrac13\), and thus Werner states are two-copy undistillable in this range. Our aim is to show that this range extends to \(\alpha = - \tfrac{1}{2}\).

\section{Contributions}\label{statements-to-be-proved}

We prove the required inequality, \eqref{eq:crw-criterion} at
\(\alpha=-\tfrac12\), in the following stronger form which is
valid for every rank and local dimensions \(d_1, d_2\).
Note that no assumption on positive-semidefiniteness, hermiticity, or normality is made on \(C\).

\begin{restatable}[Rank-constrained partial-trace inequality]{theoremA}{thmRankPT}
Let \(d_1,d_2,r\) be positive integers with \(1\le r\le d_1d_2\), and let
\(C\in M_{d_1d_2}(\C)\) have rank at most \(r\).  Then
\begin{equation}\label{eq:7-1}
\|\tr_1(C)\|_F^2+\|\tr_2(C)\|_F^2\le r\|C\|_F^2+\frac1r|\tr(C)|^2 \,.
\end{equation}
\end{restatable}

Theorem~A is proved in Section~\ref{proof-of-the-rank-constrained-partial-trace-inequality}.
By Theorem~\ref{thm:criterion-alt} follows:
\begin{restatable}{corollaryB}{corProblemFive}
The two-ququart Werner state \(\varrho(4,-\tfrac12)\) is {\em not} two-copy distillable.
\end{restatable}

In fact, we can deduce from Theorem~A also:

\begin{restatable}{corollaryC}{corAllDWerner}
For every \(d\ge2\), the Werner state \(\varrho(d,\alpha)\) is two-copy undistillable
if and only if \(\alpha\ge-\tfrac12\).
\end{restatable}

\section{Overview}\label{sec:overview}

\subsection{Dependency table}

Below we list, for every item, its provenance and dependencies.  We cite the original reference for items that predate this
note (these are collected in Section~\ref{sec:known}), and
highlight with \isnew\ for results established in this work.

\renewcommand*{\arraystretch}{1.2}
\begingroup\small
\begin{longtable}{@{}p{0.095\textwidth}p{0.27\textwidth}p{0.09\textwidth}p{0.175\textwidth}p{0.25\textwidth}@{}}
\toprule
Item & Name & Source & Depends on & Used in\\
\midrule
\endfirsthead
\toprule
Item & Name & Source & Depends on & Used in\\
\midrule
\endhead
Lem.~\ref{lem:parker} & Parker--Fillmore constant diagonal form & \cite{Parker1948,Fillmore1969} & --- & Lem.~\ref{lem:balanced}\\
Lem.~\ref{lem:svd} & singular-value decomposition & \cite{HornJohnson1991,Bhatia1997} & spectral theorem & Lem.~\ref{lem:balanced}\\
Lem.~\ref{lem:lagrange} & Lagrange (variance) identity & \cite{HardyLittlewoodPolya1934} & --- & Thm.~A, Step~6\\
Lem.~\ref{lem:rankonept} & rank-one partial-trace inequality & \cite{CostaRico2025,CostaRicoWolf2025} & \eqref{eq:2-5}, \eqref{eq:2-6}, \eqref{eq:2-11}, \eqref{eq:2-12} & Thm.~A, Steps~2 and~3\\
Lem.~\ref{lem:psd} & Theorem~A for positive matrices & \cite{Audenaert2007} & \eqref{eq:rank-trace-proof} & not used; delimits the elementary regime\\
Lem.~\ref{lem:esidentity} & four-state marginal identity & \cite{EltschkaSiewert2020} & quoted; or~\eqref{eq:2-11}, \eqref{eq:2-12} & Cor.~\ref{cor:crossedpol}\\
\midrule
Cor.~\ref{cor:crossedpol} & crossed polarization & \isnew & Lem.~\ref{lem:esidentity} & Thm.~A, Step~3\\
Lem.~\ref{lem:balanced} & balanced rank-\(r\) decomposition & \isnew & Lems.~\ref{lem:svd}, \ref{lem:parker} & Thm.~A\\
Thm.~A & rank-constrained partial-trace inequality & \isnew & Lems.~\ref{lem:balanced}, \ref{lem:rankonept}, \ref{lem:lagrange}; Cor.~\ref{cor:crossedpol} & Cors.~B and~C\\
Cor.~B & Problem~\ref{prob:five}: \(\varrho(4,-\tfrac12)\) is two-copy undistillable & \isnew & Thm.~A at \(d_1=d_2=4\), \(r=2\); Thm.~\ref{thm:criterion-alt} & --- (the target)\\
Cor.~C & Problem~\ref{prob:five-point-one} target & \isnew & Thm.~A at \(r=2\); \eqref{eq:rank-trace-proof}; Thm.~\ref{thm:criterion-alt}; separability/NPT criterion~\cite{Werner1989} & ---\\
\bottomrule
\end{longtable}
\endgroup

\subsection{Dependency graph}

The same information in graphical form; arrows point from an ingredient to the
result that uses it.  Shaded nodes are known,
framed nodes are established here.

\bigskip
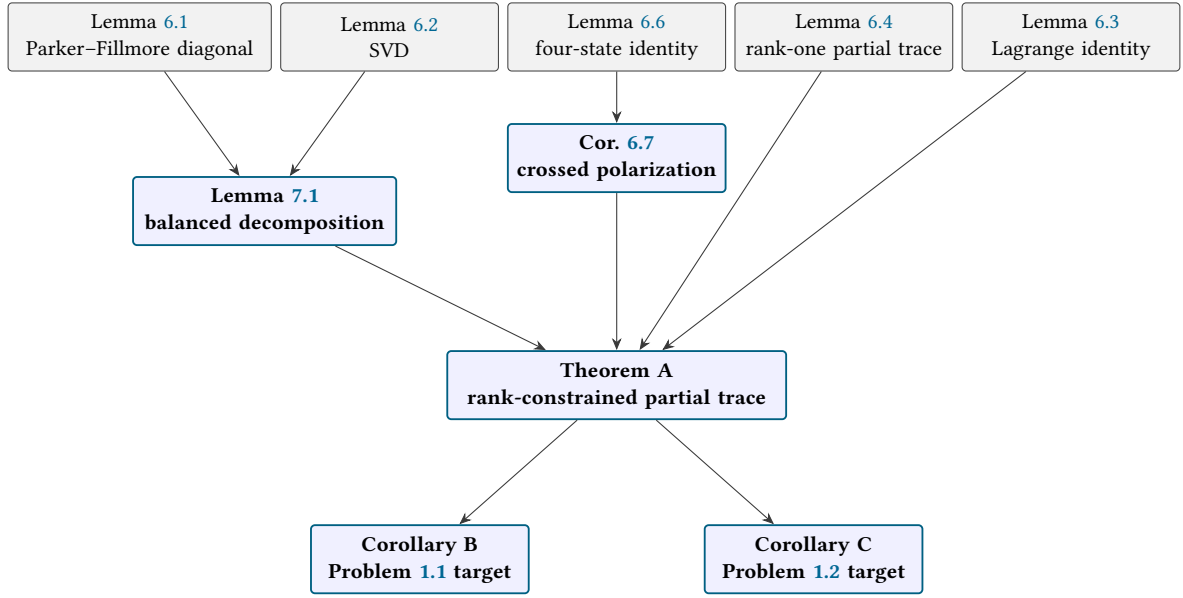
\begin{center}
\begin{tikzpicture}[
  x=1cm, y=1cm,
  kn/.style={draw=gray!55!black, fill=gray!10, rounded corners=1.5pt,
             inner sep=2.5pt, align=center, font=\scriptsize,
             text width=27mm, minimum height=9mm},
  nw/.style={draw=MidnightBlue!80!black, line width=0.7pt, fill=blue!6,
             rounded corners=1.5pt, inner sep=2.5pt, align=center,
             font=\scriptsize\bfseries, text width=27mm, minimum height=9mm},
  ar/.style={-{Stealth[length=1.6mm,width=1.3mm]}, gray!45!black, line width=0.4pt}
]
\node[kn, text width=33mm] (park)  at (-6.3, 0)    {Lemma~\ref{lem:parker}\\Parker--Fillmore diagonal};
\node[kn] (svd)   at (-3.0, 0)    {Lemma~\ref{lem:svd}\\SVD};
\node[kn] (esid)  at ( 0.0, 0)    {Lemma~\ref{lem:esidentity}\\four-state identity};
\node[nw] (cross) at ( 0.0, -1.6) {Cor.~\ref{cor:crossedpol}\\crossed polarization};
\node[kn] (ropt)  at ( 3.0, 0)    {Lemma~\ref{lem:rankonept}\\rank-one partial trace};
\node[kn] (lagr)  at ( 6.0, 0)    {Lemma~\ref{lem:lagrange}\\Lagrange identity};

\node[nw, text width=33mm] (bal)   at (-4.65,-2.3)  {Lemma~\ref{lem:balanced}\\balanced decomposition};
\node[nw, text width=43mm] (thma)  at ( 0.0,-4.6)  {Theorem~A\\rank-constrained partial trace};

\node[nw] (p5)    at (-2.6,-6.9)  {Corollary~B\\Problem~\ref{prob:five} target};
\node[nw] (alld)  at ( 2.6,-6.9)  {Corollary~C\\
Problem~\ref{prob:five-point-one} target};

\draw[ar] (svd)  -- (bal);
\draw[ar] (park) -- (bal);
\draw[ar] (esid) -- (cross);
\draw[ar] (bal)  -- (thma);
\draw[ar] (ropt) -- (thma);
\draw[ar] (cross)-- (thma);
\draw[ar] (lagr) -- (thma);
\draw[ar] (thma) -- (p5);
\draw[ar] (thma) -- (alld);
\end{tikzpicture}
\end{center}

\bigskip

\section{Conventions and definitions}\label{conventions-and-definitions}

\subsection{Index sets and matrix spaces}\label{index-sets-and-matrix-spaces}

For a positive integer \(d\), write \([d]:=\{1,2,\ldots,d\}\).
For positive integers \(m,n\), let \(M_{m,n}(\C)\) denote the vector space of \(m\times n\) complex matrices. Write \(M_n(\C):=M_{n,n}(\C)\).
The conjugate transpose of a matrix \(T\) is denoted by \(T^*\).
The ordinary transpose is denoted by \(T^T\). Entrywise complex conjugation is denoted by \(\overline T\).

\subsection{Inner products}\label{inner-products}

For column vectors \(x,y\in\C^n\) the inner product and norm are
\begin{equation}\label{eq:ip-vec}
\langle x,y\rangle:=x^*y=\sum_{i=1}^n\overline{x_i}\,y_i\,,
\qquad
\|x\|:=\sqrt{\langle x,x\rangle}=\Bigl(\sum_{i=1}^n|x_i|^2\Bigr)^{1/2}\,.
\end{equation}
For matrices \(S,T\in M_{m,n}(\C)\) the Frobenius, or Hilbert--Schmidt, inner product
and norm are
\begin{equation}\label{eq:ip-mat}
\langle S,T\rangle_F:=\tr(S^*T)=\sum_{i=1}^m\sum_{j=1}^n\overline{S_{ij}}\,T_{ij}\,,
\qquad
\|T\|_F:=\sqrt{\langle T,T\rangle_F}
=\Bigl(\sum_{i=1}^m\sum_{j=1}^n|T_{ij}|^2\Bigr)^{1/2}\,.
\end{equation}
Both are conjugate-linear in the first argument and linear in the second.
A third pairing on a direct sum of two matrix spaces, appears
in~\eqref{eq:5-2}; it is again the same one, since the pair \((R_1,R_2)\) may be read
as the block-diagonal matrix \(\operatorname{diag}(R_1,R_2)\), for which
\(\tr\bigl(\operatorname{diag}(R_1,R_2)^*\operatorname{diag}(S_1,S_2)\bigr)
=\langle R_1,S_1\rangle_F+\langle R_2,S_2\rangle_F\).  No new notation is therefore
introduced for it.

\subsection{Rank-one operators}\label{rank-one-operators}

For column vectors \(x,y\in\C^N\), the matrix \(xy^*\)
is the rank-at-most-one operator defined by \((xy^*)z=x(y^*z)\).
Its trace and Frobenius norm satisfy
\begin{align}
\tr(xy^*)&=y^*x\,, \label{eq:rank1-trace}\\
\|xy^*\|_F^2&=\|x\|^2\|y\|^2\,. \label{eq:rank1-norm}
\end{align}
Note that the left-hand side of~\eqref{eq:rank1-norm} is a matrix norm and the
right-hand side a product of vector norms; the subscript \(F\) marks the
difference.
Indeed, \(\tr(xy^*)=\sum_{k=1}^N x_k\overline{y_k}=y^*x\),
and
\begin{equation}
\begin{aligned}
\|xy^*\|_F^2
&=\tr((xy^*)^*(xy^*))\\
&=\tr(yx^*xy^*)\\
&=(x^*x)\tr(yy^*)\\
&=\|x\|^2\|y\|^2\,.
\end{aligned}
\end{equation}
\subsection{Trace cyclicity, including rectangular products}\label{trace-cyclicity-including-rectangular-products}

If \(R\in M_{m,n}(\C)\) and \(S\in M_{n,m}(\C)\), then
\begin{equation}\label{eq:2-3}
\tr(RS)=\tr(SR)\,.
\end{equation}
This is standard; see Horn and Johnson~\cite[Sec.~0.3]{HornJohnson1991} or
Bhatia~\cite[Sec.~I.1]{Bhatia1997}.

\subsection{Tensor-product coordinates}\label{tensor-product-coordinates}

Let \(\mathcal H_1=\C^{d_1}\) and \(\mathcal H_2=\C^{d_2}\), with standard orthonormal
bases \((e_a)_{a\in[d_1]}\) and \((f_b)_{b\in[d_2]}\).  We identify
\(\mathcal H_1\otimes\mathcal H_2\cong\C^{d_1d_2}\) using the ordered basis
\begin{equation}
(e_a\otimes f_b)_{a\in[d_1],\,b\in[d_2]}\,.
\end{equation}
A vector \(x\in\mathcal H_1\otimes\mathcal H_2\) is written \(x=\sum_{a=1}^{d_1}\sum_{b=1}^{d_2}x_{ab}\,e_a\otimes f_b\).
A matrix \(C\in M_{d_1d_2}(\C)\) has entries \(C_{(a,b),(c,d)}\)
with \(a,c\in[d_1]\) and \(b,d\in[d_2]\).
For \(P\in M_{d_1}(\C)\) and \(Q\in M_{d_2}(\C)\), their Kronecker product is defined by
\begin{equation}\label{eq:2-4}
(P\otimes Q)_{(a,b),(c,d)}=P_{ac}Q_{bd}\,.
\end{equation}
\subsection{Partial traces}\label{partial-traces}

For \(C\in M_{d_1d_2}(\C)\), define \(\tr_1(C)\in M_{d_2}(\C)\)
and \(\tr_2(C)\in M_{d_1}(\C)\)
by
\begin{align}
(\tr_1(C))_{bd} &:= \sum_{a=1}^{d_1}C_{(a,b),(a,d)}, \label{eq:2-5}\\
(\tr_2(C))_{ac} &:= \sum_{b=1}^{d_2}C_{(a,b),(c,b)}. \label{eq:2-6}
\end{align}
Thus \(\tr_1\) traces out the first tensor factor, while \(\tr_2\) traces out the second tensor factor.

The partial traces are adjoint to tensoring with an identity matrix: for every
\(P\in M_{d_1}(\C)\) and \(Q\in M_{d_2}(\C)\),
\begin{align}
\langle C,P\otimes I_{d_2}\rangle_F &= \langle \tr_2(C),P\rangle_F, \label{eq:2-8}\\
\langle C,I_{d_1}\otimes Q\rangle_F &= \langle \tr_1(C),Q\rangle_F. \label{eq:2-9}
\end{align}
Equations~\eqref{eq:2-8}--\eqref{eq:2-9} can be seen as the coordinate-free definition of
the partial trace: \(\tr_2\) is the unique linear map
\(M_{d_1d_2}(\C)\to M_{d_1}(\C)\) satisfying~\eqref{eq:2-8} for all \(P\), and
\(\tr_1\) the unique one satisfying~\eqref{eq:2-9} for all
\(Q\)~\cite[\S1.1.2]{Watrous2018},~\cite[\S2.4.3]{NielsenChuang2010}.

\subsection{Vectorization}\label{vectorization}

For \(X\in M_{d_1,d_2}(\C)\), define the row-major vectorization,
\begin{equation}\label{eq:2-10}
\operatorname{vec}(X)
:=
\sum_{a=1}^{d_1}\sum_{b=1}^{d_2}X_{ab}\,e_a\otimes f_b\,.
\end{equation}
This provides a linear bijection
\(\operatorname{vec}:M_{d_1,d_2}(\C)\longrightarrow\mathcal H_1\otimes\mathcal H_2\),
and it is an isometry between the two inner products of
Section~\ref{inner-products}:
\begin{equation}\label{eq:vec-isometry}
\langle\operatorname{vec}(X),\operatorname{vec}(Y)\rangle=\langle X,Y\rangle_F\,,
\qquad\text{so that}\qquad
\|\operatorname{vec}(X)\|=\|X\|_F\,.
\end{equation}
If \(x=\operatorname{vec}(X)\) and \(y=\operatorname{vec}(Y)\), then the rank-one
operator \(xy^*\) satisfies
\begin{align}
\tr_2(xy^*) &= XY^*, \label{eq:2-11}\\
\tr_1(xy^*) &= (Y^*X)^T. \label{eq:2-12}
\end{align}
Indeed,
\begin{align}
(\tr_2(xy^*))_{ac}&=\sum_b x_{ab}\overline{y_{cb}}=(XY^*)_{ac}\,,\\
(\tr_1(xy^*))_{bd}&=\sum_a x_{ab}\overline{y_{ad}}=(Y^*X)_{db}=((Y^*X)^T)_{bd}\,.
\end{align}
These are partial-trace forms of the identity \(\operatorname{vec}(ABC)=(A^T\otimes
C)\operatorname{vec}(B)\) [see ~\cite[\S2.3, Eq.~(4)]{HendersonSearle1981} for the column majorized form].

\subsection{Rank and trace}\label{rank-trace}

For every \(C\in M_N(\C)\),
\begin{equation}\label{eq:rank-trace-proof}
|\tr(C)|^2\le\rank(C)\,\|C\|_F^2\,.
\end{equation}
Indeed, for a singular-value decomposition \(C=\sum_{j=1}^s\mu_ju_jv_j^*\)
with \(s=\rank(C)\), Cauchy--Schwarz gives
\(|\tr(C)| = |\sum_{j=1}^s\mu_ju_jv_j^*|
= \sum_{j=1}^s \mu_j |u_jv_j^*|
\le\sum_j\mu_j\le\sqrt s\,\|C\|_F\).
%
%
%
%

\subsection{The paired partial-trace map}\label{sec:phi}

Define the linear map
\begin{equation}
\Phi:M_{d_1d_2}(\C)
\longrightarrow
M_{d_2}(\C)\oplus M_{d_1}(\C)
\end{equation}
by
\begin{equation}\label{eq:5-1}
\Phi(C):=\big(\tr_1(C),\tr_2(C)\big)\,.
\end{equation}
The direct sum is equipped with the inner product
\begin{equation}\label{eq:5-2}
\langle (R_1,R_2),(S_1,S_2)\rangle
:=
\langle R_1,S_1\rangle_F
+
\langle R_2,S_2\rangle_F\,.
\end{equation}
The two summands are orthogonal, so this is again the pairing~\eqref{eq:ip-mat},
read on block-diagonal matrices.  Thus
\begin{equation}\label{eq:5-3}
\|\Phi(C)\|^2
=
\|\tr_1(C)\|_F^2
+
\|\tr_2(C)\|_F^2\,.
\end{equation}
Equations~\eqref{eq:5-1} and~\eqref{eq:5-2} are definitions, and~\eqref{eq:5-3}
is the norm they induce.
The displayed norm is the genuine
norm on the target direct sum: \(\tr_1(C)\) occupies its
\(M_{d_2}(\C)\) slot and \(\tr_2(C)\) its
\(M_{d_1}(\C)\) slot.

\section{Known ingredients}\label{sec:known}

Everything in this section predates the manuscript, with the exception of
Corollary~\ref{cor:crossedpol} which is a straightforward consequence of Lemma~\ref{lem:esidentity}.

\subsection{Constant diagonals}

\begin{restatable}[Parker--Fillmore constant diagonal form \cite{Parker1948,Fillmore1969}]{lemma}{thParkerDiagonal}\label{lem:parker}
Let \(K\in M_r(\C)\), and define \(\tau:=\frac{\tr(K)}{r}\).
Then there exists a unitary matrix \(Q\in M_r(\C)\) such that
\begin{equation}
(Q^*KQ)_{ii}=\tau\quad\text{for every }i\in[r]\,;
\end{equation}
equivalently, there is an orthonormal basis \(q_1,\ldots,q_r\) of \(\C^r\) with
\(q_i^*Kq_i=\tau\) for every \(i\in[r]\).
\end{restatable}

The statement used here is~\cite[Theorem~9]{Parker1948}: for every square complex
matrix there is a unitary \(Q\) for which \(Q^*KQ\) has all diagonal entries equal;
since the trace is unitarily invariant, that common value is \(\tau=\tr(K)/r\).  The
same statement is~\cite[Corollary~2]{Fillmore1969}, with the constant likewise fixed
by the trace condition of~\cite[Theorem~2]{Fillmore1969}.

\subsection{Singular-value decomposition}

{
\begin{restatable}[Singular-value decomposition \cite{HornJohnson1991,Bhatia1997}]{lemma}{thCompactSVD}\label{lem:svd}
Let \(C\in M_N(\C)\) have rank \(s\). The case \(s=0\) is interpreted using empty families and an empty sum. Then there exist positive numbers \(\mu_1,\ldots,\mu_s>0\)
and orthonormal families
\begin{equation}
u_1,\ldots,u_s\in\C^N,
\qquad
v_1,\ldots,v_s\in\C^N
\end{equation}
such that \(C=\sum_{i=1}^s\mu_i u_i v_i^*\).
\end{restatable}
}

\subsection{Lagrange variance identity}

\begin{restatable}[Lagrange variance identity {\cite[Theorem 7]{HardyLittlewoodPolya1934}}]{lemma}{thLagrange}\label{lem:lagrange}
For real numbers \(\delta_1,\ldots,\delta_r\),
\begin{equation}\label{eq:lagrange}
r\sum_{i=1}^r \delta_i^2-\Bigl(\sum_{i=1}^r \delta_i\Bigr)^2
=\sum_{1\le i<j\le r}(\delta_i-\delta_j)^2\ \ge\ 0\,.
\end{equation}
Equality in~\eqref{eq:lagrange} is reached for
\(\delta_1=\cdots=\delta_r\).
\end{restatable}

To see this, note that both sides of~\eqref{eq:lagrange} equal
\((r-1)\sum_i\delta_i^2-2\sum_{i<j}\delta_i\delta_j\).  On the left,
\(\bigl(\sum_i\delta_i\bigr)^2=\sum_i\delta_i^2+2\sum_{i<j}\delta_i\delta_j\); on the
right, expanding each square gives
\(\sum_{i<j}(\delta_i-\delta_j)^2=(r-1)\sum_i\delta_i^2-2\sum_{i<j}\delta_i\delta_j\),
since each \(\delta_i^2\) occurs in exactly \(r-1\) of the \(\binom r2\) terms.

\subsection{Rank-one partial-trace identities}

\begin{restatable}[Rank-one partial-trace inequality \cite{CostaRico2025,CostaRicoWolf2025}]{lemma}{thRankOnePT}\label{lem:rankonept}
For every \(x,y\in\C^{d_1}\otimes\C^{d_2}\),
one has
\begin{equation}
\|\Phi(xy^*)\|^2
\le
\|x\|^2\|y\|^2+|y^*x|^2\,.
\end{equation}
Equivalently,
\begin{equation}
\|\tr_1(xy^*)\|_F^2
+
\|\tr_2(xy^*)\|_F^2
\le
\|xy^*\|_F^2+|\tr(xy^*)|^2\,.
\end{equation}
\end{restatable}

The inequality is \cite[Proposition~1]{CostaRico2025}; it is also the \(\gamma=2\)
case of \cite[Proposition~11, Eq.~(28)]{CostaRicoWolf2025}.

Lemma~\ref{lem:rankonept} can be seen as follows:
Let \(F_A\) and \(F_B\) denote the swap operators acting on the pairs of \(\C^{d_1}\) and \(\C^{d_2}\) factors respectively.
Then,
\begin{equation}\label{eq:rankone-certificate}
\|x\|^2\|y\|^2+|y^*x|^2-\|\Phi(xy^*)\|^2
=\bigl\langle x\otimes y,\bigl[(I-F_A)\otimes(I-F_B)\bigr]x\otimes y\bigr\rangle
=4\bigl\|(P_A^-\otimes P_B^-)(x\otimes y)\bigr\|^2\ \ge\ 0\,,
\end{equation}
where \(P_A^-=\tfrac12(I-F_A)\) and \(P_B^-=\tfrac12(I-F_B)\) are the antisymmetric
projections.

\begin{restatable}[Theorem~A for positive matrices \cite{Audenaert2007}]{lemma}{thPSDbound}\label{lem:psd}
Let \(C\in M_{d_1d_2}(\C)\) be positive semidefinite.  Then, with no hypothesis on
its rank,
\begin{equation}\label{eq:psd-rankfree}
\|\tr_1(C)\|_F^2+\|\tr_2(C)\|_F^2\le\|C\|_F^2+(\tr(C))^2\,.
\end{equation}
If moreover \(\rank(C)\le r\), then
\begin{equation}\label{eq:psd-thmA}
\|\tr_1(C)\|_F^2+\|\tr_2(C)\|_F^2\le r\|C\|_F^2+\frac1r(\tr(C))^2\,,
\end{equation}
which is the upper bound of Theorem~A.
\end{restatable}

Inequality~\eqref{eq:psd-rankfree} is the \(q=2\) case of positive-matrix entropy
subadditivity~\cite{Audenaert2007}.
We note that Lemma~\ref{lem:psd} is not used in the proof of
Theorem~A.

\subsection{A four-state identity}\label{sec:es-identity}

The identity of Eltschka and Siewert~\cite{EltschkaSiewert2020}
below converts a trace inner product of the marginals of two rank-one operators into a complementary expression with terms exchanged.
What Theorem~A eventually uses is its direct-sum consequence,
Corollary~\ref{cor:crossedpol}.

\begin{lemma}[Four-state marginal identity {\cite[Eq.~(13)]{EltschkaSiewert2020}}]\label{lem:esidentity}
In the original form of Eltschka and Siewert, for all
\(\psi,\chi,\phi,\zeta\in\C^{d_1}\otimes\C^{d_2}\),
\begin{equation}\label{eq:es-bilinear}
\tr\!\bigl[\tr_1(\psi\chi^*)\,\tr_1(\phi\zeta^*)\bigr]
=\tr\!\bigl[\tr_2(\psi\zeta^*)\,\tr_2(\phi\chi^*)\bigr]\,,
\end{equation}
a bilinear trace identity.
\end{lemma}

Equivalently, through the Frobenius inner product, for all
\(u,v,s,t\in\C^{d_1}\otimes\C^{d_2}\),
\begin{align}
\bigl\langle \tr_1(uv^*),\tr_1(st^*)\bigr\rangle_F
 &= \bigl\langle \tr_2(tv^*),\tr_2(su^*)\bigr\rangle_F\,,
 \label{eq:es-component-1}\\
\bigl\langle \tr_2(uv^*),\tr_2(st^*)\bigr\rangle_F
 &= \bigl\langle \tr_1(tv^*),\tr_1(su^*)\bigr\rangle_F\,.
 \label{eq:es-component-2}
\end{align}
The two forms are equivalent by replacing \((\psi,\chi,\phi,\zeta)=(v,u,s,t)\)
in~\eqref{eq:es-bilinear} and using that the Frobenius product conjugates its first
argument, e.g. \(\langle\tr_1(uv^*),\tr_1(st^*)\rangle_F=\tr[\tr_1(vu^*)\,\tr_1(st^*)]\).

For the benefit of the reader we prove Lemma~\ref{lem:esidentity} directly from the
vectorization identities~\eqref{eq:2-11}--\eqref{eq:2-12}.

\begin{proof}[Proof of Lemma~\ref{lem:esidentity}]
It suffices to prove~\eqref{eq:es-component-1};
equation~\eqref{eq:es-component-2} is the same identity with the two tensor factors
interchanged, and both are equivalent to the bilinear form~\eqref{eq:es-bilinear}.

By the vectorization bijection~\eqref{eq:2-10}, write
\(u=\operatorname{vec}(U)\), \(v=\operatorname{vec}(V)\),
\(s=\operatorname{vec}(S)\), \(t=\operatorname{vec}(T)\)
with \(U,V,S,T\in M_{d_1,d_2}(\C)\).

Evaluate the left-hand side of~\eqref{eq:es-component-1}:
by~\eqref{eq:2-12}, \(\tr_1(uv^*)=(V^*U)^T\) and \(\tr_1(st^*)=(T^*S)^T\), and since
transposition preserves the Frobenius inner product,
\begin{equation*}
\langle\tr_1(uv^*),\tr_1(st^*)\rangle_F=\langle V^*U,T^*S\rangle_F=\tr(U^*VT^*S)\,.
\end{equation*}
Now evaluate the right-hand side of~\eqref{eq:es-component-1}:
by~\eqref{eq:2-11}, \(\tr_2(tv^*)=TV^*\) and \(\tr_2(su^*)=SU^*\), so
\begin{equation*}
\langle\tr_2(tv^*),\tr_2(su^*)\rangle_F=\langle TV^*,SU^*\rangle_F=\tr(VT^*SU^*)\,.
\end{equation*}
By trace cyclicity~\eqref{eq:2-3}, \(\tr(U^*VT^*S)=\tr(VT^*SU^*)\), so the two sides
agree, which is~\eqref{eq:es-component-1}.
This ends the proof.
\end{proof}

\begin{corollary}[Crossed polarization]\label{cor:crossedpol}
For \(u,v,s,t\in\C^{d_1}\otimes\C^{d_2}\),
\begin{equation}\label{eq:crossedpol}
\bigl\langle\Phi(uv^*),\Phi(st^*)\bigr\rangle
=\bigl\langle\Phi(tv^*),\Phi(su^*)\bigr\rangle\,.
\end{equation}
\end{corollary}

\begin{proof}
By the definition~\eqref{eq:5-2} of the direct-sum inner product, the left-hand
side of~\eqref{eq:crossedpol} is the sum of the left-hand sides
of~\eqref{eq:es-component-1} and~\eqref{eq:es-component-2}, and the right-hand side
is the sum of their right-hand sides.  Adding the two therefore gives the claim.
This ends the proof.
\end{proof}

\begin{remark}
 Setting \((u,v,s,t)=(x_i,y_i,x_j,y_j)\), Corollary~\ref{cor:crossedpol} reads
\begin{equation}\label{eq:crossedpol-xy}
\bigl\langle\Phi(x_i y_i^*),\Phi(x_j y_j^*)\bigr\rangle
=\bigl\langle\Phi(y_j y_i^*),\Phi(x_j x_i^*)\bigr\rangle\,,
\end{equation}
the form used in the proof of Theorem~A.
\end{remark}

\section{The balanced rank decomposition}\label{singular-value-decomposition-and-a-balanced-rank-decomposition}

We now state a main ingredient of the proof, a decomposition of $C$
such that the diagonal is constant and the Gram matrices corresponding to the left and right factors coincide.

\begin{lemma}[Balanced rank-\(r\) decomposition]\label{lem:balanced}
Let \(C\in M_N(\C)\)
satisfy \(\rank(C)\le r\le N\).
Then there exist vectors \(x_1,\ldots,x_r\in\C^N\) and
\(y_1,\ldots,y_r\in\C^N\), and a Hermitian positive semidefinite matrix \(G=(g_{ij})_{i,j=1}^r\in M_r(\C)\)
such that all of the following hold:
\begin{align}
C &= \sum_{i=1}^r x_i y_i^*\,, \label{eq:4-2}\\
x_i^*x_j &= y_i^*y_j = g_{ij}
  \qquad\text{for every }i,j\in[r]\,, \label{eq:4-3}\\
y_i^*x_i &= \frac{\tr(C)}{r}
  \qquad\text{for every }i\in[r]\,. \label{eq:4-4}
\end{align}
\end{lemma}
\begin{proof}

Let \(s:=\rank(C)\le r\).
Take a singular-value decomposition from Lemma~\ref{lem:svd}: \(C=\sum_{i=1}^s\mu_i u_i v_i^*\).
Extend \(u_1,\ldots,u_s\) to an orthonormal family \(u_1,\ldots,u_r\), and extend \(v_1,\ldots,v_s\) to an orthonormal family \(v_1,\ldots,v_r\). Define \(\mu_{s+1}=\cdots=\mu_r=0\).
Let
\begin{equation}
U=(u_1\ \cdots\ u_r),
\qquad
V=(v_1\ \cdots\ v_r)
\end{equation}
be the corresponding \(N\times r\) isometries, and let \(\Sigma:=\operatorname{diag}(\mu_1,\ldots,\mu_r)\).
Then \(C=U\Sigma V^*\).
Define
\begin{equation}\label{eq:4-5}
X:=U\Sigma^{1/2},
\qquad
Y:=V\Sigma^{1/2}\,.
\end{equation}
Then \(C=XY^*\) and \(X^*X=\Sigma=Y^*Y\).

Set \(K:=Y^*X\in M_r(\C)\).
By rectangular trace cyclicity \eqref{eq:2-3} it holds that
\begin{equation}\label{eq:4-8}
\tr(K)
=
\tr(Y^*X)
=
\tr(XY^*)
=
\tr(C)\,.
\end{equation}
By Lemma~\ref{lem:parker}, there exists a unitary \(Q\in M_r(\C)\) such that every diagonal entry of \(Q^*KQ\) is
\begin{equation}\label{eq:4-9}
\tau:=\frac{\tr(C)}{r}\,.
\end{equation}

Now let \(\widetilde X:=XQ\) and \(\widetilde Y:=YQ\) and note that
\begin{equation}
 C = \widetilde X\widetilde Y^* = XQQ^*Y^* =XY^*\,, \label{eq:4-10}
\end{equation}
using \(QQ^*=I\).  Also define
\begin{equation}
 G := \widetilde X^*\widetilde X=Q^*\Sigma Q =\widetilde Y^*\widetilde Y\,, \label{eq:4-11}
\end{equation}
using \(X^*X=\Sigma=Y^*Y\).  Note that \(G\succeq0\), being a Gram matrix.

Let \(x_i\) be the \(i\)-th column of \(\widetilde X\), and \(y_i\) the \(i\)-th column
of \(\widetilde Y\).  Then \(\widetilde X\widetilde Y^*=\sum_{i=1}^r x_iy_i^*\), so
by~\eqref{eq:4-10}, \(C=\sum_{i=1}^r x_iy_i^*\), which is~\eqref{eq:4-2}.

For~\eqref{eq:4-3}, read off the entries of \(G\):
\begin{equation}\label{eq:gram-entries}
(\widetilde X^*\widetilde X)_{ij}=x_i^*x_j\,,
\qquad
(\widetilde Y^*\widetilde Y)_{ij}=y_i^*y_j
\qquad\text{for all }i,j\in[r]\,.
\end{equation}
By~\eqref{eq:4-11} the two matrices both equal \(G\), so
\(x_i^*x_j=y_i^*y_j=g_{ij}\), which is~\eqref{eq:4-3}.

Finally,
\begin{equation}
y_i^*x_i=(\widetilde Y^*\widetilde X)_{ii}=(Q^*Y^*XQ)_{ii}=(Q^*KQ)_{ii}=\tau\,,
\end{equation}
shows \eqref{eq:4-4}.
This ends the proof.
\end{proof}

\section{Proof of Theorem A}\label{proof-of-the-rank-constrained-partial-trace-inequality}

\thmRankPT*
\begin{proof}
The strategy is as follows:
By~\eqref{eq:5-3}, the left-hand side of~\eqref{eq:7-1} equals
\(\|\Phi(C)\|^2\).
Expanding \(C\) as a balanced rank-\(r\)
decomposition we can write \(\Phi(C)=\sum_i\Phi(x_i y_i^*)\).
Its squared norm \(\|\Phi(C)\|^2\) can be written as a double sum of diagonal and off-diagonal terms,
which we bound separately from above.
The resulting upper bound plus an additional
sum of squares expression then equals the right-hand
side of~\eqref{eq:7-1}, proving the assertion.

\subsection{Step 1: balanced decomposition}\label{step-1-balanced-decomposition}

Let \(N:=d_1d_2\).
Apply the balanced decomposition in
Lemma~\ref{lem:balanced} to obtain vectors
\(x_1,\ldots,x_r,\ y_1,\ldots,y_r\in\C^N\) and a Hermitian positive semidefinite
matrix \(G=(g_{ij})_{i,j=1}^r\) such that
\begin{equation}\label{eq:balanced-props}
C=\sum_{i=1}^r x_i y_i^*,
\qquad
x_i^*x_j=y_i^*y_j=g_{ij}\quad(1\le i,j\le r),
\qquad
y_i^*x_i=\tau\quad(1\le i\le r)\,,
\end{equation}
where
\begin{equation}\label{eq:7-5}
\tau:=\frac{\tr(C)}{r}\,.
\end{equation}
We also define \(\delta_i:=g_{ii}=\|x_i\|^2=\|y_i\|^2\ge0\).

By linearity of \(\Phi\),
\begin{equation}\label{eq:7-8}
\Phi(C)=\sum_{i=1}^r\Phi(x_i y_i^*)\,.
\end{equation}
Then,
\begin{equation}\label{eq:double-sum}
\|\Phi(C)\|^2
=\sum_{i,j=1}^r\bigl\langle\Phi(x_i y_i^*),\Phi(x_j y_j^*)\bigr\rangle\,.
\end{equation}
Our aim in steps~2 and~3 in then to
bound the diagonal and off-diagonal terms appearing
in $\|\Phi(C)\|^2$.

\subsection{Step 2: diagonal terms}\label{step-1-diagonal-terms}

By Lemma~\ref{lem:rankonept}, \eqref{eq:balanced-props}, and the definition of \(\delta_i\),
\begin{equation}\label{eq:7-9}
\begin{aligned}
\|\Phi(x_i y_i^*)\|^2
&\le
\|x_i\|^2\|y_i\|^2+|y_i^*x_i|^2\\
&=\delta_i^2+|\tau|^2\,.
\end{aligned}
\end{equation}
\subsection{Step 3: cross terms}\label{step-2-cross-terms}

Fix \(i\ne j\). By Corollary~\ref{cor:crossedpol}, respectively the form given in Eq.~\eqref{eq:crossedpol-xy},
\begin{equation}\label{eq:7-10}
\bigl\langle\Phi(x_i y_i^*),\Phi(x_j y_j^*)\bigr\rangle
=
\left\langle
\Phi(y_jy_i^*),
\Phi(x_jx_i^*)
\right\rangle\,.
\end{equation}
By Cauchy--Schwarz in the direct-sum Hilbert space,
\begin{equation}\label{eq:7-11}
\bigl|\bigl\langle\Phi(x_i y_i^*),\Phi(x_j y_j^*)\bigr\rangle\bigr|
=\bigl|\bigl\langle\Phi(y_jy_i^*),\Phi(x_jx_i^*)\bigr\rangle\bigr|
\le
\|\Phi(y_jy_i^*)\|\,
\|\Phi(x_jx_i^*)\|\,.
\end{equation}

Now consider the rank-one operator \(y_jy_i^*\) appearing on the rhs of \eqref{eq:7-11}.
By \eqref{eq:rank1-trace}, \eqref{eq:rank1-norm}, and
then the definition of \(\delta_i\) in \eqref{eq:balanced-props},
\begin{align}
\|y_jy_i^*\|_F^2 &= \|y_j\|^2\|y_i\|^2 = \delta_j\delta_i\,, \label{eq:7-12}\\
\tr(y_jy_i^*) &= y_i^*y_j = g_{ij}\,. \label{eq:7-13}
\end{align}
By Lemma~\ref{lem:rankonept} therefore
\begin{equation}\label{eq:7-14}
\|\Phi(y_jy_i^*)\|^2
\le
\delta_i\delta_j+|g_{ij}|^2\,.
\end{equation}
The same computation for \(x_jx_i^*\) gives
\begin{equation}\label{eq:7-15}
\|\Phi(x_jx_i^*)\|^2
\le
\delta_i\delta_j+|g_{ij}|^2\,.
\end{equation}
Combining \eqref{eq:7-11}, \eqref{eq:7-14}, and \eqref{eq:7-15},
\begin{equation}\label{eq:7-16}
\begin{aligned}
\bigl|\bigl\langle\Phi(x_i y_i^*),\Phi(x_j y_j^*)\bigr\rangle\bigr|
&\le\|\Phi(y_jy_i^*)\|\,\|\Phi(x_jx_i^*)\|\\
&\le\sqrt{\delta_i\delta_j+|g_{ij}|^2}\;\sqrt{\delta_i\delta_j+|g_{ij}|^2}\\
&=\delta_i\delta_j+|g_{ij}|^2\,.
\end{aligned}
\end{equation}

These inequalities \eqref{eq:7-9} and \eqref{eq:7-16} finish the upper bound on the lhs of \eqref{eq:7-1}.

\subsection{Step 4: a combined upper bound}\label{step-3-sum-all-terms}

Now expand the squared norm of $\Phi(C)$. Using \eqref{eq:double-sum},
\begin{equation}\label{eq:7-17}
\begin{aligned}
\|\Phi(C)\|^2
&=\sum_{i,j=1}^r\bigl\langle\Phi(x_i y_i^*),\Phi(x_j y_j^*)\bigr\rangle\,\\
&=
\sum_{i=1}^r\|\Phi(x_i y_i^*)\|^2
+
2\operatorname{Re}\sum_{1\le i<j\le r}
 \bigl\langle\Phi(x_i y_i^*),\Phi(x_j y_j^*)\bigr\rangle\,.
\end{aligned}
\end{equation}
Using the upper bounds of \eqref{eq:7-9} and \eqref{eq:7-16},
\begin{equation}\label{eq:7-18}
\begin{aligned}
\|\Phi(C)\|^2
&\le
\sum_{i=1}^r(\delta_i^2+|\tau|^2)
+
2\sum_{1\le i<j\le r}
(\delta_i\delta_j+|g_{ij}|^2)\\
&=
\left(\sum_{i=1}^r \delta_i\right)^2
+r|\tau|^2
+2\sum_{1\le i<j\le r}|g_{ij}|^2\,.
\end{aligned}
\end{equation}
Above, we have regrouped the terms involving $\delta_i$ by
\(\sum_i\delta_i^2+2\sum_{i<j}\delta_i\delta_j = \bigl(\sum_i\delta_i\bigr)^2\).

\subsection{\texorpdfstring{Step 5: express \(\|C\|_F^2\) through the common Gram matrix}{Step 5: express \textbackslash\textbar C\textbackslash\textbar\_F\^{}2 via the Gram matrix}}\label{step-4-express-c_f2-through-the-common-gram-matrix}

As in the balanced
rank-\(r\) decomposition of Lemma~\ref{lem:balanced},
let \(\widetilde X\) have columns \(x_1\ \cdots\ x_r\) and
\(\widetilde Y\) have columns \(y_1\ \cdots\ y_r\).
Then
\begin{align}
 C = \widetilde X\widetilde Y^* \,, \quad\quad
 G = \widetilde X^*\widetilde X =\widetilde Y^*\widetilde Y\,,
\end{align}
where \(G=(g_{ij})_{ij}\).  Therefore
\begin{equation}\label{eq:7-19}
\begin{aligned}
\|C\|_F^2
&=\tr(C^*C)
=\tr(\widetilde Y\widetilde X^*\widetilde X\widetilde Y^*)
=\tr(\widetilde X^*\widetilde X\widetilde Y^*\widetilde Y)
=\tr(G^2)\,.
\end{aligned}
\end{equation}

Since \(G\) is Hermitian, \(g_{ji}=\overline{g_{ij}}\).
Hence
\begin{equation}\label{eq:7-20}
\begin{aligned}
\tr(G^2)
&=
\sum_{i=1}^r\sum_{j=1}^r g_{ij}g_{ji}\\
&=
\sum_{i=1}^r \delta_i^2
+
2\sum_{1\le i<j\le r}|g_{ij}|^2\,.
\end{aligned}
\end{equation}
Thus
\begin{equation}\label{eq:7-21}
\|C\|_F^2
=
\sum_{i=1}^r \delta_i^2
+
2\sum_{1\le i<j\le r}|g_{ij}|^2\,.
\end{equation}
\subsection{Step 6: the remaining term is a sum of squares}\label{step-5-the-remaining-error-is-a-sum-of-squares}

Taking the modulus squared of~\eqref{eq:7-5} gives \(|\tau|^2=|\tr(C)|^2/r^2\), and hence
\begin{equation}\label{eq:7-22}
\frac1r|\tr(C)|^2
=r|\tau|^2\,.
\end{equation}

Let us now show that \eqref{eq:7-1} holds.
To this end, subtract its lhs from its rhs.
Using \eqref{eq:7-21}, \eqref{eq:7-22}
and the upper bound on \(||\Phi(C)||^2\) of \eqref{eq:7-18},
we have
\begin{equation}\label{eq:7-23}
\begin{aligned}
&\quad r\|C\|_F^2+\frac1r|\tr(C)|^2
- \|\tr_1(C)\|_F^2
-
\|\tr_2(C)\|_F^2
\\
&=r\|C\|_F^2
+\frac1r|\tr(C)|^2
-\|\Phi(C)\|^2\\
&\ge
r\left(
\sum_i \delta_i^2
+2\sum_{i<j}|g_{ij}|^2
\right)
+r|\tau|^2
-
\left[
\left(\sum_i \delta_i\right)^2
+r|\tau|^2
+2\sum_{i<j}|g_{ij}|^2
\right]\\
\end{aligned}
\end{equation}
The last expression can then be regrouped using the Lagrange identity [Lemma~\ref{lem:lagrange}],
\begin{align}
&r\sum_i \delta_i^2
-
\left(\sum_i \delta_i\right)^2
+
2(r-1)\sum_{i<j}|g_{ij}|^2 \\
& \quad=\quad
\sum_{i<j}(\delta_i-\delta_j)^2
+
2(r-1)\sum_{i<j}|g_{ij}|^2 \quad \ge0\,.
\end{align}
This is a clearly non-negative expression, and thus proves
inequality \eqref{eq:7-1}.
This ends the proof.

\end{proof}

\section{Back to distillability}

We can now answer the first question.

\probFive*

 The answer is:

\corProblemFive*

\begin{proof}
Theorem~A with \(d_1=d_2=4\) and \(r=2\) states that every
\(C\in M_{16}(\C)\) of rank at most two satisfies
\(\|\tr_1(C)\|_F^2+\|\tr_2(C)\|_F^2\le2\|C\|_F^2+\tfrac12|\tr(C)|^2\).  This is
condition~\eqref{eq:crw-criterion} of
Theorem~\ref{thm:criterion-alt} at \(d=4\), \(\alpha=-\tfrac12\).
It follows that
\(\varrho(4,-\tfrac12)\) is two-copy undistillable.
This ends the proof.
\end{proof}

We can also answer the second main question:

\probGeneral*

The answer is:

\corAllDWerner*

\begin{proof}
We proceed by cases:

\noindent {\bf Case 1: \(-\tfrac{1}{d}\le \alpha\).}
\quad \(\varrho(d,\alpha)\) is separable~\cite{Werner1989}
and thus two-copy undistillable.

\smallskip
For \(\alpha<-\tfrac{1}{d}\), \(\varrho(d,\alpha)\) is NPT~\cite{Werner1989} and we
distinguish two further cases:

\smallskip
\noindent {\bf Case 2: \(-\tfrac12\le\alpha<-\tfrac{1}{d}\).}
\quad Note that this is only relevant for \(d\ge3\).
Let
\(C\in M_{d^2}(\C)\) have rank at most two, put \(t:=-\alpha\in(0,\tfrac12]\),
\begin{equation}
 a:=\|C\|_F^2\,, \quad b:=|\tr(C)|^2\,,
\end{equation}
By Theorem~A with \(d_1=d_2=d\) and \(r=2\),
\(\|\Phi(C)\|^2\le2a+b/2\), so it suffices to show that
\begin{equation}\label{eq:merged-bound-proof}
2a+\frac b2\ \le\ \frac at+tb\,,
\end{equation}
for all \(t:=-\alpha\in(0,\tfrac12]\).
This indeed holds due to,
\begin{equation}
 \bigl(\tfrac at+tb\bigr)-\bigl(2a+\tfrac b2\bigr)=(1-2t)\Bigl(\tfrac at-\tfrac b2\Bigr) \geq 0\,.
\end{equation}
Above we used that the first factor is nonnegative
due to \(t\le\tfrac12\).
The second factor is nonnegative due to
\eqref{eq:rank-trace-proof}, \(b\le2a\),
which gives \(bt\le2at\le a\le2a\), i.e.\ \(\tfrac{a}{t}\ge \tfrac{b}{2}\).  Hence~\eqref{eq:merged-bound-proof} holds.

Combining \eqref{eq:merged-bound-proof} with Theorem~A gives
\begin{equation}
 \|\Phi(C)\|^2\le \frac{a}{t}+tb=\tfrac1{|\alpha|}\|C\|_F^2+|\alpha|\,|\tr(C)|^2
\end{equation}
which is
condition~\eqref{eq:crw-criterion} of Theorem~\ref{thm:criterion-alt}.  Hence
\(\varrho(d,\alpha)\) is two-copy undistillable
for  \(-\tfrac12\le\alpha\).

\smallskip
\noindent {\bf Case 3: \(\alpha<-\tfrac12\).}
Then \(\varrho(d,\alpha)\) is
one-copy distillable~\cite{Horodecki1998}, and one-copy
distillability trivially implies distillability at every larger number of copies;
hence \(\varrho(d,\alpha)\) is two-copy distillable.

This ends the proof.
\end{proof}
As a consequence, the one and two-copy distillability regions
of $\varrho(d,\alpha)$ coincide.


\begin{thebibliography}{99}
\small

\bibitem{Werner1989}
R.~F. Werner,
``Quantum states with Einstein--Podolsky--Rosen correlations admitting a hidden-variable model,''
\emph{Physical Review A} \textbf{40}, 4277--4281 (1989).
\href{https://doi.org/10.1103/PhysRevA.40.4277}{doi:10.1103/PhysRevA.40.4277}.

\bibitem{Horodecki1998}
M.~Horodecki, P.~Horodecki, and R.~Horodecki,
``Mixed-state entanglement and distillation: Is there a `bound' entanglement in nature?,''
\emph{Physical Review Letters} \textbf{80}, 5239--5242 (1998).
\href{https://doi.org/10.1103/PhysRevLett.80.5239}{doi:10.1103/PhysRevLett.80.5239}.

\bibitem{Peres1996}
A.~Peres,
``Separability criterion for density matrices,''
\emph{Physical Review Letters} \textbf{77}, 1413--1415 (1996).
\href{https://doi.org/10.1103/PhysRevLett.77.1413}{doi:10.1103/PhysRevLett.77.1413}.

\bibitem{HorodeckiReduction1999}
M.~Horodecki and P.~Horodecki,
``Reduction criterion of separability and limits for a class of distillation protocols,''
\emph{Physical Review A} \textbf{59}, 4206--4216 (1999).
\href{https://doi.org/10.1103/PhysRevA.59.4206}{doi:10.1103/PhysRevA.59.4206}.

\bibitem{DurCiracLewensteinBruss2000}
W.~D\"ur, J.~I.~Cirac, M.~Lewenstein, and D.~Bru\ss,
``Distillability and partial transposition in bipartite systems,''
\emph{Physical Review A} \textbf{61}, 062313 (2000).
\href{https://doi.org/10.1103/PhysRevA.61.062313}{doi:10.1103/PhysRevA.61.062313};
arXiv:quant-ph/9910022.

\bibitem{Vianna2006}
R.~O. Vianna and A.~C. Doherty,
``Distillability of Werner states using entanglement witnesses and robust semidefinite programs,''
\emph{Physical Review A} \textbf{74}, 052306 (2006).
\href{https://doi.org/10.1103/PhysRevA.74.052306}{doi:10.1103/PhysRevA.74.052306}.

\bibitem{Pankowski2010}
\L{}.~Pankowski, M.~Piani, M.~Horodecki, and P.~Horodecki,
``A Few Steps More Towards NPT Bound Entanglement,''
\emph{IEEE Transactions on Information Theory} \textbf{56}(8), 4085--4100 (2010).
\href{https://doi.org/10.1109/TIT.2010.2050810}{doi:10.1109/TIT.2010.2050810}.

\bibitem{EltschkaSiewert2020}
C.~Eltschka and J.~Siewert,
``Joint Schmidt-type decomposition for two bipartite pure quantum states,''
\emph{Physical Review A} \textbf{101}, 022302 (2020).
\href{https://doi.org/10.1103/PhysRevA.101.022302}{doi:10.1103/PhysRevA.101.022302};
arXiv:1907.07976.

\bibitem{FiveOpen2022}
P.~Horodecki, \L{}.~Rudnicki, and K.~\.{Z}yczkowski,
``Five Open Problems in Quantum Information Theory,''
\emph{PRX Quantum} \textbf{3}, 010101 (2022).
\href{https://doi.org/10.1103/PRXQuantum.3.010101}{doi:10.1103/PRXQuantum.3.010101}; arXiv:2002.03233.

\bibitem{CostaRico2025}
P.~Costa Rico,
``New Partial Trace Inequalities and Distillability of Werner States,''
\emph{Letters in Mathematical Physics} \textbf{115}, 47 (2025).
\href{https://doi.org/10.1007/s11005-025-01935-y}{doi:10.1007/s11005-025-01935-y}; arXiv:2310.05726.

\bibitem{CostaRicoWolf2025}
P.~Costa Rico and M.~M. Wolf,
``Partial Trace Relations Beyond Normal Matrices,''
arXiv:2507.18278 [quant-ph] (2025).
\href{https://doi.org/10.48550/arXiv.2507.18278}{doi:10.48550/arXiv.2507.18278}.

\bibitem{Fillmore1969}
P.~A. Fillmore,
``On Similarity and the Diagonal of a Matrix,''
\emph{The American Mathematical Monthly} \textbf{76}(2), 167--169 (1969).
\href{https://doi.org/10.1080/00029890.1969.12000162}{doi:10.1080/00029890.1969.12000162}.

\bibitem{Parker1948}
W.~V. Parker,
``Sets of Complex Numbers Associated with a Matrix,''
\emph{Duke Mathematical Journal} \textbf{15}(3), 711--715 (1948).
\href{https://doi.org/10.1215/S0012-7094-48-01560-9}{doi:10.1215/S0012-7094-48-01560-9}.

\bibitem{HendersonSearle1981}
H.~V. Henderson and S.~R. Searle,
``The vec-permutation matrix, the vec operator and Kronecker products: a review,''
\emph{Linear and Multilinear Algebra} \textbf{9}(4), 271--288 (1981).
\href{https://doi.org/10.1080/03081088108817379}{doi:10.1080/03081088108817379}.

\bibitem{Audenaert2007}
K.~M.~R. Audenaert,
``Subadditivity of \(q\)-entropies for \(q>1\),''
\emph{Journal of Mathematical Physics} \textbf{48}, 083507 (2007).
\href{https://doi.org/10.1063/1.2771542}{doi:10.1063/1.2771542}.

\bibitem{HardyLittlewoodPolya1934}
G.~H. Hardy, J.~E. Littlewood, and G.~P\'olya,
\emph{Inequalities}, Cambridge University Press, Cambridge (1934);
2nd ed.\ (1952).

\bibitem{NielsenChuang2010}
M.~A. Nielsen and I.~L. Chuang,
\emph{Quantum Computation and Quantum Information}, 10th anniversary ed.,
Cambridge University Press, Cambridge (2010).
\href{https://doi.org/10.1017/CBO9780511976667}{doi:10.1017/CBO9780511976667}.

\bibitem{Watrous2018}
J.~Watrous,
\emph{The Theory of Quantum Information},
Cambridge University Press, Cambridge (2018).
\href{https://doi.org/10.1017/9781316848142}{doi:10.1017/9781316848142}.

\bibitem{HornJohnson1991}
R.~A. Horn and C.~R. Johnson,
\emph{Topics in Matrix Analysis}, Cambridge University Press, Cambridge (1991).

\bibitem{Bhatia1997}
R.~Bhatia,
\emph{Matrix Analysis}, Springer, New York (1997).

\end{thebibliography}
\end{document}